\documentclass{llncs}[10point]

\usepackage[lncs,amsthm,autoqed]{pamas}
\usepackage{txfonts}
\usepackage{arydshln}
\usepackage{wrapfig}
\usepackage{algorithm,algorithmic} 

\newcommand{\set}[1]{\{#1\}}

\newcommand\eat[1]{}

\usepackage{tikz}
\usetikzlibrary{arrows}

\usepackage{mathrsfs}

\newcommand{\B}{$\mathcal{B}$}

\newcommand{\midd}{\mathbin{:}}

\usepackage{calrsfs}

\DeclareMathOperator*{\gdot}{\dot{>}}
\DeclareMathOperator*{\gdoteq}{\dot{\geq}}

\newcommand{\pref}{\succsim\xspace}
\newcommand{\Pref}[1][]{
	\ifthenelse{\equal{#1}{}}{\mathrel \succsim}{\mathop{\succsim_{#1}}}
}                                          
\newcommand{\sPref}[1][]{                  
	\ifthenelse{\equal{#1}{}}{\mathrel \succ}{\mathop{\succ_{#1}}}
}                                          
\newcommand{\Indiff}[1][]{                 
	\ifthenelse{\equal{#1}{}}{\mathrel \sim}{\mathop{\sim_{#1}}}
}
\newcommand{\prefset}[1][]{\ifthenelse{\equal{#1}{}}{\mathcal{\succsim}}{\mathcal{\succsim}_{#1}}}

\title{Existence of Stability in Hedonic Coalition Formation Games
\thanks{This material is based upon work supported by the Deutsche Forschungsgemeinschaft under the grant BR-2312/6-1 (within the European Science Foundation's EUROCORES program LogICCC). Thanks to Felix Brandt and Hans Georg Seedig for helpful comments.}}

\author{%
Haris Aziz\inst{1} \and
  Florian Brandl\inst{2}}
\institute{%
  Department of Informatics,
  Technische Universit\"at M\"unchen, 
  85748 M\"unchen, Germany \\
  \email{aziz@in.tum.de}
  \and
  Department of Mathematics,
  Technische Universit{\"a}t M{\"u}nchen,
  85748 M\"unchen, Germany \\
  \email{f.brandl@mytum.de}}

\begin{document}

\maketitle

\begin{abstract}
	In this paper, we examine \emph{hedonic coalition formation games} in which each player's preferences over partitions of players depend only on the members of his coalition. We present three main results in which restrictions on the preferences of the players 
guarantee the existence of stable partitions for various notions of stability. The preference restrictions pertain to 	
\emph{top responsiveness} and \emph{bottom responsiveness} which model optimistic and pessimistic behavior of players respectively. The existence results apply to natural subclasses of \emph{additive separable hedonic games} and \emph{hedonic games with \B-preferences}. It is also shown that our existence results cannot be strengthened to the case of stronger known stability concepts.
\end{abstract}

%
%
%



\section{Introduction}

In many models of multiagent interaction such as roommate matching and exchange of discrete goods, deviations from one outcome to another can cycle and it may well be possible that no stable outcome is guaranteed. This leads to one of the most fundamental questions in game theory: \emph{what are the necessary and sufficient conditions for the existence of stable outcomes?} This question has been examined extensively by researchers working in market design, multiagent systems, and operations research. 
We address this question in the context of coalition formation games in which outcomes are partitions of the players. We focus on \emph{hedonic coalition formation games} in which each player's preferences over partitions depend only on the members of his coalition. Hedonic games are a rich and versatile class of coalition formation games which also encapsulate various stable matching scenarios~\cite[see e.g., ][]{BKS01a,BoJa02a,Cech08a,Hajd06a}.

In game theory and multiagent systems, understanding the conditions under which systems and social outcomes are guaranteed to be in equilibrium is a fundamental research problem.
In this paper, we advance the state of the art on existence results for hedonic games.
We strengthen the recently introduced stability concept \emph{strong Nash stability}~\citep{Kara11a} to \emph{strict strong Nash stability} and show that \emph{top responsiveness} and \emph{mutuality}---conditions different from ones in \citep{Kara11a}---are sufficient for the existence of a strictly strong Nash stable partition in any hedonic game. 
The result applies to natural subclasses of \emph{additive separable hedonic games}~\citep{BoJa02a}. 
 It is also shown that top responsiveness and mutuality together do not guarantee the non-emptiness of the set of \emph{perfect} partitions---a natural concept stronger than strict strong Nash stability. 

We then consider a recently introduced property of hedonic games called \emph{bottom refuseness}~\citep{SuSu10a} which we will refer to as \emph{bottom responsiveness}.  
A new stability notion called \emph{strong individual stability} is formulated which is stronger than both core stability and individual stability. It is shown that bottom responsiveness guarantees the existence of a strong individually stable partition. Also, the combination of \emph{strong bottom responsiveness} and mutuality guarantees the existence of a strong Nash stable partition. 
Our results concerning bottom responsive games cannot be strengthened to any stronger known stability concept. They also apply to \emph{`aversion to enemies' games} introduced in \citep{DBHS06a}.

\paragraph{Outline}

In Section~\ref{sec:related}, we present the backdrop of our results and discuss related work. We then introduce hedonic games and the stability concepts considered for these games in Section~\ref{sec:prel}. The relationships between the stability concepts are expounded and clarified in Section~\ref{sec:rel}. We then proceed to Sections~\ref{sec:tr} and \ref{sec:br} in which the main results are presented. 
Section~\ref{sec:tr} concerns hedonic games satisfying top responsiveness whereas in Section~\ref{sec:br}, existence results concerning bottom responsive games are presented. 
In Section~\ref{sec:app}, well-studied subclasses of hedonic games such as additive separable hedonic games and hedonic games with \B-preferences are considered and it is shown how existence results apply to these games. Finally, we conclude the discussion in Section~\ref{sec:conc}.

\section{Related Work}\label{sec:related}

Identifying sufficient and necessary conditions for the existence of stability in coalition formation has been active area of research. Perhaps the most celebrated result in this field is the existence of a (core) stable matching for the stable marriage problem via the Gale-Shapley algorithm~\citep{GaSh62a}. 
Later, \citet{BKS01a} proved that if a hedonic game satisfies a condition called \emph{weak top coalition property}, then the core is non-empty. 
\citet{BKS01a} also showed that for various restrictions over preferences, stability is still not guaranteed. 

In another important paper, \citet{BoJa02a} formalized Nash stability and individual stability in the context of hedonic games and presented a number of sufficient conditions for the existence of various stability concepts. For instance, they showed that \emph{symmetric additively separable preferences} guarantee the existence of a Nash stable partition. A hedonic game is additively separable if each player has a cardinal value for every other player and the player's utility in a partition is the sum of his values for the players in his coalition. 
The strict core and core is also non-empty for \emph{`appreciation of friends'} and  \emph{`aversion to enemies'} games respectively---two simple classes of additively separable games~\citep{DBHS06a}.

\citet{AlRe04a} proposed a natural preference restriction called 
top responsiveness which is based on the idea that players value other players on how they could complement them in research teams.  They showed that there exists an algorithm called the \emph{Top Covering Algorithm} which finds a core stable partition for top responsive hedonic games. The Top Covering Algorithm can be seen as a generalization of \emph{Gale's Top Trading Cycle algorithm}~\cite{ShSc74a}. 
 \citet{DiSu06a,DiSu07a} simplified the Top Covering Algorithm and proved that top responsiveness implies non-emptiness of the strict core and
if mutuality is additionally satisfied, then a Nash stable partition exists.

In a follow-up paper, \citet{SuSu10a} introduced bottom refuseness in an analogous way to top responsiveness.  They showed that for hedonic games satisfying bottom refuseness,
 the \emph{Bottom Avoiding Algorithm} returns a core stable partition. \citet{SuSu10a} noted that `appreciation of friends' and `aversion to enemies' games satisfy top responsiveness and bottom responsiveness respectively, thereby explaining the results in ~\citep{DBHS06a}.
 
Very recently, \citet{Kara11a} proposed a new stability concept called \emph{strong Nash stability} which is stronger than Nash stability and core stability combined. He showed that strong-Nash is non-empty if the \emph{weak top choice property} (stronger than the weak top coalition property) is satisfied or if preferences are `\emph{descending separable}'. 
We will prove three different results in which natural restrictions on the player preferences guarantee the existence of stable partitions where stability is strong Nash stability or its generalization or variant.

\section{Hedonic Games \& Stability Concepts}\label{sec:prel}

In this section, we review the terminology, notation, and concepts related to hedonic games.

\paragraph{Hedonic games}

A \emph{hedonic coalition formation game} is a pair $(N,\pref)$ where $N$ is a set of players and $\pref$ is a \emph{preference profile} which specifies for each player $i\in N$ the preference relation $ \succsim_i$, a reflexive, complete and transitive binary relation on set $\mathcal{N}_i=\{S\subseteq N \midd i\in S\}$.
$S\succ_iT$ denotes that $i$ strictly prefers $S$ over $T$ and $S\sim_iT$ that $i$ is indifferent between coalitions $S$ and $T$. A \emph{partition} $\pi$ is a partition of players $N$ into disjoint coalitions. By $\pi(i)$, we denote the coalition in $\pi$ which includes player $i$.

\paragraph{Stability Concepts}

We present the various stability concepts for hedonic games. 
Nash stability, strict core stability, Pareto optimality, core stability, and individual rationality are classic stability concepts. Individual stability was formulated in \cite{BoJa02a}. Strong Nash stability was introduced by \citet{Kara11a} and perfect partitions were considered in \citep{ABH11b}.
In this paper, we also introduce strict strong Nash stability and strong individual stability which imply strong Nash stability and core stability respectively.


\begin{itemize}
\item 	A partition $\pi$ is \emph{individually rational (IR)} if no player has an incentive to become alone, i.e., for all $i\in N$, $\pi(i) \succsim_i  \{i\}$. 

\item A partition is \emph{perfect} if each player is in one of his most preferred coalition~\citep{ABH11b}. 
\item A partition is  \emph{Nash stable (NS)} if no player can benefit by 
	moving from his coalition to another (possibly empty) coalition $T$.
\item A partition is \emph{individually stable (IS)} if no player can
	benefit by moving from his coalition to another existing (possibly empty) coalition $T$  while not making the members of $T$ worse off. 
\item A coalition $S \subseteq N$ \emph{blocks} a partition $\pi$, if each
player $i \in S$ strictly prefers $S$ to his current coalition $\pi(i)$ in
the partition $\pi$. 
A partition which admits no blocking coalition is said to be in the \emph{core (C)}. 			

\item A coalition $S \subseteq N$ \emph{weakly blocks} a partition $\pi$,
if each player $i \in S$ weakly prefers $S$ to $\pi(i)$ and there exists at least one player $j \in S$ who strictly prefers $S$ to his current coalition $\pi(j)$. A partition which admits no weakly blocking coalition is in the \emph{strict core (SC)}. 

\item A partition~$\pi$ is \emph{Pareto optimal (PO)} if there is no partition~$\pi'$ with $\pi'(j)\Pref[j]\pi(j)$ for all players~$j$ and $\pi'(i) \sPref[i]\pi(i)$ for at least one player~$i$.
	 
\item For partition $\pi$, $\pi' \neq \pi$ is called \emph{reachable} from $\pi$ by movements of players $H \subseteq N$, denoted by $\pi \mathop{\stackrel{H}{\rightarrow}} \pi'$, if 
$\forall i,j\in N\setminus H, i\neq j: \pi(i) = \pi(j) \Leftrightarrow \pi'(i) = \pi'(j).$

A subset of players $H \subseteq N, H\neq \emptyset$ \emph{strong Nash blocks} $\pi$ if a partition $\pi' \neq \pi$ exists with $\pi \stackrel{H}{\rightarrow} \pi'$ and 
$\forall i\in H: \pi'(i) \succ_i \pi(i).$ 

If a partition $\pi$ is not strong Nash blocked by any set $H \subseteq N$, $\pi$ is called \emph{strong Nash stable (SNS)}~\citep{Kara11a}.

\item 

A subset of players $H \subseteq N, H\neq \emptyset$ \emph{weakly Nash blocks} $\pi$ if a partition $\pi' \neq \pi$ exists with $\pi \stackrel{H}{\rightarrow} \pi'$, $\forall i\in H: \pi'(i) \succsim_i \pi(i)$ and $\exists i\in H: \pi'(i) \succ_i \pi(i)$. 

A partition which admits no weakly Nash blocking coalition is said to satisfy \emph{strict strong Nash stability (SSNS)}.

\item A non-empty set of players $H\subseteq N$ is \emph{strongly individually blocking} a partition $\pi$, if a partition $\pi'$ exists such that: 

\begin{enumerate}
\item $\pi \stackrel{H}{\rightarrow} \pi'$ (as for SNS), 
\item $\forall i \in H: \pi'(i) \succ_i \pi(i)$, and
\item $\forall j \in \pi'(i) \text{ for some } i \in H:  \pi'(j) \succsim_j \pi(j)$.
\end{enumerate}


A partition for which no strongly individually blocking set exists is \emph{strongly individually stable (SIS)}.\footnote{SIS is a natural intermediate stability concept which is implied by strong Nash stability and strict core stability  respectively and it also implies individual stability and core stability.}

\end{itemize}

	\begin{figure}[h]

		\begin{center}
			\scalebox{0.9}{
			\begin{tikzpicture}
				\tikzstyle{pfeil}=[->,>=angle 60, shorten >=1pt,draw]
				\tikzstyle{onlytext}=[]

				\node[onlytext] (PT) at (4,6) {Perfect};
				\node[onlytext] (SSNS) at (4,4.5) {SSNS};
					\node[onlytext] (SNS) at (2,3) {SNS};
				\node[onlytext] (NS) at (2,1.5) {NS};
		
		\node[onlytext] (SIS) at (4,1.5) {SIS};
				\node[onlytext] (SC) at (6,3) {SC};
				\node[onlytext] (IS) at (2,0) {IS};
				\node[onlytext] (PO) at (6,1.5) {PO};
				\node[onlytext] (IR) at (4,-1.5) {IR};
				\node[onlytext] (C) at (6,0) {C};

\draw[pfeil] (IS) to (IR);
\draw[pfeil] (C) to (IR);
\draw[pfeil] (SSNS) to (SC);
\draw[pfeil] (SSNS) to (SNS);
\draw[pfeil] (PT) to (SSNS);
\draw[pfeil] (SNS) to (NS);
\draw[pfeil] (SIS) to (IS);
\draw[pfeil] (SC) to (SIS);
\draw[pfeil] (SIS) to (C);
\draw[pfeil] (SNS) to (SIS);
				\draw[pfeil] (NS) to (IS);
				\draw[pfeil] (SC) to (PO);

			\end{tikzpicture}
			
}

			\end{center}
			\caption{Inclusion relationships between stability concepts for hedonic games. For e.g, every NS partition is also IS. NS, SC, PO, C and IR are classic stability concepts. IS was formulated in \cite{BoJa02a}; SNS in \cite{Kara11a}; and perfect partitions in \cite{ABH11b}. We also introduce SSNS and SIS in this paper.}
			\label{fig:relations}
			
			\end{figure}
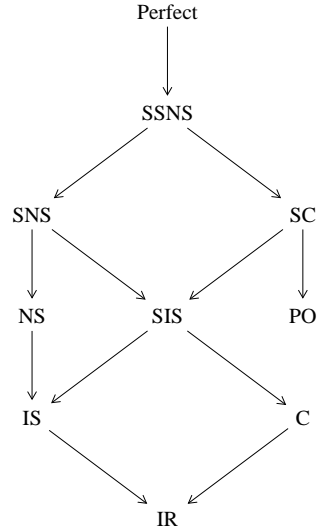

Depending on the context, we will utilize abbreviations like SIS, SNS, SSNS, IS etc. either for adjectives (for e.g. IS for individually stable) or for nouns (for e.g. IS for individual stability).

\section{Relations between Stability Concepts}\label{sec:rel}

In this section, we will explore and clarify the inclusion relationships between the stability concepts.
The inclusion relationships between stability concepts are depicted in Figure~1.

\begin{proposition}\label{prop:stability-relations}
	Strict core stability implies strong individual stability which implies individual stability and also core stability.
	\end{proposition}

\begin{proof}
	
	Strict core stability implies strong individual stability. Assume that a partition $\pi$ is strict core stable but not strong individually stable. Then, there exists a coalition $S\subseteq N$ such that $S\notin \pi$ and each player in $S$ is at least as happy as in $\pi$ and one player in $S$ is strictly happier than in $\pi$. But this means that $\pi$ is not strict core stable.
	
	Strong individual stability trivially implies individual stability.
	
	Finally, we show that strong individual stability implies core stability. Assume that a partition $\pi$ is strong individually stable but not core stable. Then there is a core deviating coalition $S$. But this would mean that each player 
$i\in S$ is strictly better off than in $\pi(i)$. But this means that $\pi$ is not strong individually stable. 
This complete the proof.
\end{proof}


Strong Nash stability as introduced by ~\citet{Kara11a} is quite a strong stability notion as seen by the following simple proposition.

\begin{proposition}\label{prop:SNS-is-strong}
Strong Nash stability implies Nash stability and also core stability. 

Furthermore, even if a partition is both strict core stable and Nash stable, it is not necessarily strong Nash stable.
\end{proposition}
\begin{proof}
	The first statement follows from the definitions of the stability concepts and was already pointed out by ~\citet{Kara11a}. In fact, it can also easily be shown that SNS implies SIS. If a partition is SNS, then there is no strong Nash blocking set. This implies that there does not exist any strongly individually blocking set. 	
	
	
We now show that even if a partition is both strict core stable and Nash stable, it is not necessarily strong Nash stable.	
The following example shows a game, that admits a strict core and Nash stable partition but no strong Nash stable partition.

Let $(N, \pref)$ be a game with $N=\lbrace 1,2,3,4 \rbrace$ and the preference profile specified as follows:
\begin{align*}
&\lbrace 1,2 \rbrace \succ_1 \{1,4\} \succ_1 \lbrace 1 \rbrace  \succ_1  ...\\
&\lbrace 2,3 \rbrace \succ_2 \lbrace 1,2 \rbrace \succ_2 \lbrace 2 \rbrace  \succ_2  ...\\
&\lbrace 3,4 \rbrace \succ_3 \{2,3\} \succ_3 \lbrace 3 \rbrace  \succ_3   ...\\
&\lbrace 1,4 \rbrace \succ_4 \lbrace 3,4 \rbrace \succ_4 \lbrace 4 \rbrace  \succ_4  ...
\end{align*}

It is easy to check, that the partitions $ \pi = \lbrace \lbrace 1,2 \rbrace, \lbrace 3,4 \rbrace \rbrace$ and $ \pi' = \lbrace \lbrace 1,4 \rbrace, \lbrace 2,3 \rbrace \rbrace$ are both (even strictly) core stable and Nash stable. But neither of them is strong Nash stable since $ \lbrace2,4 \rbrace$ is blocking $\pi$ and $\pi'$ is blocked by $ \lbrace 1,3 \rbrace$. Obviously any partition containing a coalition with 3 or more players is not even Nash stable, since each player prefers being alone to any coalition with more than $2$ players. Also $  \lbrace \lbrace 1,3 \rbrace, \lbrace 2,4 \rbrace \rbrace$ is not even Nash stable, since it is not individually rational.

Nash and core stability prevent single players from moving to another (possibly empty) coalition or several players forming a new coalition respectively. In the given partitions $\pi$ and $\pi'$, it is possible for a pair of players to improve by switching coalitions and therefore prevent $\pi$ and $\pi'$ from being strong Nash stable.
\end{proof}

In the next proposition, we show that although strong Nash stability is a strong stability concept, it does not imply strict core stability nor Pareto optimality. 

\begin{proposition}\label{prop:SNS-is-weak}
	Strict strong Nash stability implies strong Nash stability, strict core stability 
	and Pareto optimality. 
	
	On the other hand, strong Nash stability does not imply strict core stability nor Pareto optimality. 
\end{proposition}
\begin{proof}
	Strict strong Nash stability trivially implies strong Nash stability. 
Strict strong Nash stability also implies strict core stability. 
	If a partition is strong strong Nash stable, there exists no new coalition $H$, in which each player at least as happy and one player is strictly better off. Therefore, the partition is also strict core stable.

	Now, we will show that strong Nash stability does not imply strict core stability nor Pareto optimality. 
	Since, it is well-known that strict core stability implies Pareto optimality, it is sufficient to show that strong Nash stability does not imply Pareto optimality.

	Strong Nash stability does not imply Pareto optimality. 
	Consider the following four-player hedonic game:

\begin{itemize}
	        \item[] $\set{1,2} \sim_1 \set{1,3} \sim_1 \set{1,4} \succ_1 \cdots$
			\item[] $\set{1,2}\sim_2\set{2,3}\sim_2\set{2,4}\succ_2 \cdots$ 
	             \item[] $\set{2,3}\sim_3\set{3,4}\succ_3 \cdots$  
	\item[] $\set{1,4}\sim_4 \set{2,4}\succ_4 \set{3,4} \succ_4 \cdots $
\end{itemize}
	
	Then, the partition $\{\{1,2\}, \{3,4\}\}$ is strong Nash stable. However it is Pareto dominated by $\{\{2,3\}, \{1,4\}\}$.
	\end{proof}


In the next sections, we will present the central results of the paper. 


\section{Top responsiveness}\label{sec:tr}
\label{sec:exists}
 Top responsiveness~\citep{AlRe04a,DiSu06a,DiSu07a} and bottom responsiveness~\citep{SuSu10a} are natural restrictions that are imposed on the individual preferences and not on the whole preference profile. The idea is that a player's preference for a coalition depends on the best and worst subcoalitions respectively. In this section, we present a result that a partition fulfilling SSNS exists for hedonic games satisfying top responsiveness and an additional property called mutuality (with respect to top responsiveness).

\paragraph{Top responsiveness}

Top responsiveness is based on \emph{choice sets}---sets of players which each player wants to be with. 
Let $Ch(i,S)$---the \emph{choice sets} of player $i$ in coalition $S$---be defined as follows:
$$Ch(i,S)=\{S'\subseteq S \midd  (i\in S') \wedge (S'\succsim_i S'' ~\forall S''\subseteq S)\}.$$


A game satisfies \emph{top responsiveness} if for each $i\in N$,  the following three conditions hold:

\begin{enumerate}

\item for each $X\in \mathcal{N}_i$, $|Ch(i,X)|=1$, (we denote by $ch(i,X)$ the unique maximal set of player $i$ on $X$ under $\pref_i$),   
\item for each pair $X,Y\in \mathcal{N}_i$, $X\succ_i Y$ if $ch(i,X)\succ_i ch(i,Y)$; 
\item for each pair $X,Y\in \mathcal{N}_i$, $X\succ_i Y$ if $ch(i,X)= ch(i,Y)$ and $X\subset Y$.
\end{enumerate}

A hedonic game satisfying top responsiveness additionally satisfies \emph{mutuality} if
$$\forall i,j \in N, X \in \mathcal{N}_i \cap \mathcal{N}_j: i\in ch(j,X) \Leftrightarrow j\in ch(i,X).$$ 
We will also specify different notions of mutuality with respect to hedonic games satisfying bottom responsiveness and also separable hedonic games. When the context is clear, we will refer to the condition simply as mutuality.

\begin{example}\label{example:tr-mutual}
	Let $(N,\pref)$ be a game with $N=\lbrace 1,2,3 \rbrace$ and the preference profile specified as follows:
	\begin{align*}
	&\lbrace 1,2 \rbrace \succ_1 \lbrace 1,2,3 \rbrace  \succ_1  \lbrace 1 \rbrace  \succ_1 \lbrace 1,3 \rbrace  \\
	&\lbrace 1,2,3 \rbrace \succ_2 \lbrace 1,2 \rbrace \sim_2 \lbrace 2,3 \rbrace  \succ_2 \lbrace 2 \rbrace \\
	&\lbrace 2,3 \rbrace \succ_3 \lbrace 1,2,3 \rbrace  \succ_3 \lbrace 3 \rbrace \succ_3 \lbrace 1,3 \rbrace 
	\end{align*}
	Then, $(N,\pref)$ satisfies top responsiveness and mutuality. 
\end{example}

We are now in a position to present our first result.


\begin{theorem}\label{th:ssns-top}
	Top responsiveness and mutuality together guarantee the existence of an 
	SSNS partition.
\end{theorem}

We prove Theorem~\ref{th:ssns-top} by showing that if a hedonic game satisfies top responsiveness and mutuality, then the Top Covering Algorithm of \citep{AlRe04a,DiSu06a,DiSu07a} returns an SSNS partition. 
 Therefore, we identify conditions different than the ones identified by \citet{Kara11a} for which strong Nash stability is guaranteed. 
Since SSNS is stronger than SNS (Proposition~\ref{prop:SNS-is-weak}) which in 
turn is stronger than even the combination of Nash stability and strict core stability (Proposition~\ref{prop:SNS-is-strong}), Theorem~\ref{th:ssns-top} simultaneously strengthens the result in \citep{DiSu06a} and \citep{DiSu07a} in which it was shown that top responsiveness and mutuality together guarantee the existence of a Nash stable and strict core partition.

It can also be proved that Theorem~\ref{th:ssns-top} is optimal in the sense that it does not extend to perfect partitions. To be precise, we show that top responsiveness and mutuality together do not guarantee the existence of a perfect partition.

\begin{proposition}\label{th:perfect-no-top-mutual}
	Top responsiveness and mutuality together do not guarantee the existence of a perfect partition.
\end{proposition}
\begin{proof}
By counter example. In the game in Example~\ref{example:tr-mutual}, top responsiveness and mutuality are satisfied but no perfect partition exists.
\end{proof}

Now that we have stated Theorem~\ref{th:ssns-top} and its complementing Proposition~\ref{th:perfect-no-top-mutual}, we will present the proof of Theorem~\ref{th:ssns-top}.


\paragraph{Proof of Theorem~\ref{th:ssns-top}}

Firstly, we need additional definitions and a description of the Top Covering Algorithm.
For each $X \subseteq N $, we denote by $\backsim_X$ the relation on $X \times X $ where $i \backsim_X j$ if and only if $ j\in ch(i,X)$. In this case $j$ is called a \emph{neighbor} of $i$ in $X$. Note that in 
the preference profiles satisfies top responsiveness mutuality, then $\backsim_X$ is a symmetric relation.\\
The connected component $CC(i,X)$ of $i$ with respect to $X$ is defined as follows:
$$CC(i,X) = \{k\in X \midd \exists j_1, \dots , j_l \in X: i=j_1 \backsim_X \dots \backsim_X j_l=k \}.$$

If $j \in CC(i,X)$, $j$ is called \emph{reachable} from $i$ in $X$. Also note that $CC(j,X) \subseteq CC(i,X) $ if $j$ is reachable from $i$ and if mutuality is satisfied, then the following holds: 
 $\forall X\subseteq N$, $i,j \in X: i\in CC(j,X) \Leftrightarrow j\in CC(i,X)$.\\

Now we are ready to present the simplified Top Covering Algorithm provided by \citet{DiSu06a,DiSu07a}, adapted to the notation defined above. The algorithm is specified as Algorithm~\ref{alg:TCA}.

The following lemma will be used in the proof to Theorem \ref{th:ssns-top}. 

\begin{lemma}\label{claim_ch}
Let $(N,\pref)$ be a game satisfying top responsiveness and mutuality and $\pi$ be the partition resulting by applying the simplified top covering algorithm to it. Then
\begin{align*}
\forall i\in N: ch(i,N) \subseteq \pi(i)
\end{align*}
\end{lemma}
\begin{proof}
	First we show by induction over the iterations of the algorithm that $ch(i,R^k)=ch(i,N) \; \forall i \in R^k,\; k=1,2,...$. 
For $k=1$, this is obviously true, because $R^1=N$. 
Assume by induction, that $ch(i,R^k)=ch(i,N) \; \forall i \in R^k$. 
Let $i'$ be the player selected in the $k$-th iteration of Step \ref{step:select} 
and $j \in R^{k+1}$. 
Therefore $j\notin ch(i,R^k) \; \forall i \in CC(i',R^k)$. 
Because of mutuality $i \notin ch(j,R^k) \; \forall i \in CC(i',R^k)$. 
So $ch(j,R^k) \subseteq R^{k+1}$ and therefore $ch(j,R^{k+1}) = ch(j,R^k) = ch(j,N)$.
	
	Now take an arbitrary player $i \in N $ and denote by $k$ the iteration of the algorithm in which $i$ was added to his coalition, i.e. $i \in S^k$. Let $i'$ be the player selected in the $k$-th iteration of Step \ref{step:select}, so $i \in CC(i',R^k)$. Because of mutuality, $CC(i',R^k) = CC(i,R^k) $ and clearly $ch(i,R^k) \subseteq CC(i,R^k)$. From above, we know that $ch(i,N) = ch(i,R^k) \subseteq CC(i,R^k) =S^k = \pi(i)$.
\end{proof}

\begin{algorithm}[htb]
  \caption{Top Covering Algorithm}
  \label{alg:TCA}
  \textbf{Input:} A hedonic game $(N,\pref)$ satisfying top responsiveness.

  \begin{algorithmic}[1] 
\STATE $R^1 \leftarrow N$; $\pi \leftarrow \emptyset$.

\FOR{$k=1$ to $|N|$}
\STATE \label{step:select} Select $i\in R^k$ such that $|CC(i,R^k)| \leq |CC(j,R^k)|$ for each $j\in R^k$.
\STATE $S^k\leftarrow  CC(i,R^k)$; $\pi \leftarrow  \pi \cup \lbrace S^k \rbrace$;  and $R^{k+1} \leftarrow  R^k \setminus S^k$
 \IF{$R^{k+1} = \emptyset$}
\RETURN $\pi$
\ENDIF
\ENDFOR
\RETURN $\pi$
 \end{algorithmic}
\end{algorithm}
We note here that Lemma~\ref{claim_ch} may not hold, if mutuality is violated.


As shown by \citet{DiSu06a,DiSu07a} the resulting partition of the simplified top covering algorithm is strict core stable as well as Nash stable if preferences as mutual. We are now ready to present the proof of Theorem~\ref{th:ssns-top}.


\begin{proof}
Let $\pi $ be the resulting partition and suppose it is not strictly strong Nash stable. Then a pair $(H,\pi')$ exists where $H\subseteq N$ is the set of deviators and $\pi' $ is the partition resulting after the deviation, i.e. $\pi \stackrel{H}{\rightarrow} \pi'$.
Firstly, by Lemma~\ref{claim_ch}, $ch(i,N) \subseteq CC(i,N)$ 
$\forall i\in N$. Since $H$ is a coalition blocking strict strong Nash stability, the following holds:
\begin{align*}
&\forall i\in H: \pi'(i) \succsim_i \pi(i) \quad and\\
& \exists j\in H: \pi'(j) \succ_j \pi(j).
\end{align*}

Now consider the player $j$, who is better off in his new coalition $\pi'(j)$. Assume that $\pi(j) \cap \pi'(j) \subseteq H$, which means only deviators in $ \pi(j) \cap \pi'(j)$. For $i \in  \pi(j) \cap \pi'(j)$: $ ch(i,\pi'(i)) \succsim_i ch(i,\pi(i))$, since $i \in H$ by assumption. We also know that $ch(i,\pi(i)) = ch(i,N)$ by Lemma~\ref{claim_ch}. Therefore, for $i \in  \pi(j) \cap \pi'(j)$: $ ch(i,\pi'(i)) \succsim_i ch(i,\pi(i)) = ch(i,N)$. Because of uniqueness of choice sets in the definition of top responsiveness, $ch(i,\pi'(i)) = ch(i,N)$. So $ch(i,N) \subseteq \pi(i)\cap \pi'(i) = \pi(j) \cap \pi'(j)$. 
\begin{align*}
\Longrightarrow \forall i \in \pi(j) \cap \pi'(j): (\pi'(j) \cap \pi(j)) \succsim_i \pi'(j). 
\end{align*}

Due to assumption $\pi(j) \cap \pi'(j) \subseteq H$, the following holds:
\begin{align*}
&\forall i \in \pi(j) \cap \pi'(j): (\pi(j) \cap \pi'(j)) \succsim_i \pi'(j) \succsim_i \pi(j) = \pi(i) \quad \&\\
&(\pi(j) \cap \pi'(j)) \succsim_j \pi'(j) \succ_j \pi(j)
\end{align*}

So $\pi(j) \cap \pi'(j)$ would be a coalition blocking strict core stability, but \citet{DiSu07a} proved that $\pi$ as produced by the simplified top covering algorithm has to be strict core stable. Therefore $\pi'(j) \cap \pi(j) \nsubseteq H $ and there is at least one non-deviator in $\pi(j) \cap \pi'(j)$. Let us call this player $i'$. 

Now take a look at the players in $ \pi(j) \setminus \pi'(j)$. Note that this is not an empty set, because otherwise $\pi'(j) \supset \pi(j) \Longrightarrow \pi'(j) \precsim_j \pi(j)$. If one of them is not in $H$, then he was in the same coalition as $i'$ in $\pi$, namely $\pi(j)$, and is now in a different, which is not consistent with $\pi \stackrel{H}{\rightarrow} \pi'$. So $ (\pi(j) \setminus \pi'(j)) \subseteq H$. Because $\pi(j)$ is a connected component, at least one player $k$ in $\pi(j) \setminus \pi'(j)$ has a friend $l$ in $\pi'(j)$, meaning they are in each other's choice sets and as mentioned $l \in ch(k,N) \subseteq \pi(k)$. But now $l \notin \pi'(k)$ and therefore $ \pi'(k) \prec_k \pi(k)$ which contradicts $k$ being a deviator.
\end{proof}


\section{Bottom responsiveness}\label{sec:br}

In this section, we present the central results concerning hedonic games which satisfy bottom responsiveness.

\paragraph{Bottom responsiveness} 
Bottom responsiveness is a restriction on the preferences of each player in a hedonic game which models conservative or pessimistic agents. 
In contrast to top responsiveness, bottom responsiveness is based on \emph{avoid sets}---sets of players which each player wants to avoid having in his coalition.	
	
	For any player $i\in N$ and $S\in \mathcal{N}_i$, $Av(i,S)$---the set of \emph{avoid sets} of player $i$ in coalition $S$---is defined as follows:
	$$Av(i,S)=\{S'\subseteq S\midd (i\in S') \wedge  (S'\precsim_i S'' ~\forall S''\subseteq S) \}.$$
	A game satisfies \emph{bottom responsiveness} if for each $i\in N$, the following conditions hold: 
	
	\begin{enumerate}
	
\item for each pair $X,Y\in \mathcal{N}_i$, $X\succ_i Y$ if $X'\succ_i Y'$
for each $X'\in Av(i,X)$ and each $Y'\in Av(i,Y)$; and
\item  for each $i\in N$ and $X,Y\in \mathcal{N}_i$, $Av(i,X)\cap Av(i,Y)\neq \emptyset$ and $|X|\geq |Y|$ implies $X\succsim_i Y$.

\end{enumerate}

A hedonic game $(N,\pref)$ satisfies \emph{strong bottom responsiveness} if it is bottom responsive and if for each $i\in N$ and $X\in \mathcal{N}_i$, $|Av(i,X)|=1$. 
By $av(i,X)$, we denote the unique minimal set of player $i$ on $X$ under $\pref_i$.
The strong part of bottom responsiveness is analogous to Property 1 in the definition of top responsiveness.
	A hedonic game $(N,\pref)$ satisfying strong bottom responsiveness additionally satisfies \emph{mutuality} if for all $i,j\in N$, and $X$ such that $i,j\in X$, $i\in av(j,X)$ if and only if $j\in av(i,X)$.  

\begin{example}\label{example:br}
Let $(N,\pref)$ be a game with $N= \{1,2,3\}$ and the preference profile specified as follows:
\begin{align*}
\{1,3\} \succ_1 \{1\} \succ_1 \{1,2,3\} \succ_1 \{1,2\}\\
\{2,3\} \succ_2 \{2\} \succ_2 \{1,2,3\} \succ_2 \{1,2\}\\
\{1,2,3\} \succ_3 \{1,3\} \sim_3 \{2,3\} \succ_3 \{3\}
\end{align*}

Then, $(N,\pref)$ satisfies strong bottom responsiveness and also mutuality (with respect to strong bottom responsiveness).
\end{example}

For bottom responsive games, we prove that an SIS partition is guaranteed to exist even in the absence of mutuality. 

\begin{theorem}\label{th:bottom-sis}
	Bottom responsiveness guarantees the existence of an SIS partition.
\end{theorem}

As a corollary, a core stable partition and an individually stable partition is guaranteed to exist. Previously, it was only known that the core is non-empty for bottom responsive games~\citep{SuSu10a}. In contrast to the result by \citet{SuSu10a}, the proof of Theorem~\ref{th:bottom-sis} does not require the Bottom Avoiding Algorithm. 
We associate with each IR partition a vector of coalition sizes in decreasing order. It is then shown via lexicographic comparisons between the corresponding vectors that arbitrary deviations between partitions are acyclic. 
With an additional natural constraint, even SNS is guaranteed (Theorem~\ref{th:bottom-SNS}). 

\begin{theorem}\label{th:bottom-SNS}
	Strong bottom responsiveness and mutuality together guarantee the existence of an SNS partition.
	
\end{theorem}

We point out that Theorem~\ref{th:bottom-sis} cannot be extended any further to take care of strict core stability and Theorem~\ref{th:bottom-SNS} cannot be extended to SSNS. The reason is that \emph{symmetric `aversion to enemies' games}---a subclass of strong bottom responsive games which satisfy mutuality---may not admit a strict core stable partition~\citep[Example 4, ][]{DBHS06a}.

Now that we have stated our results concerning hedonic games satisfying bottom responsiveness, we sketch the proofs.

\paragraph{Proof of Theorem~\ref{th:bottom-sis}}

For the use of further proofs, we introduce an ordering relation on the partitions. The definition will also apply to the proof of Theorem~\ref{th:bottom-SNS}.

\begin{definition}\label{def:ordering}
Let $N= \{1,...,n\}$ be a set of players and $\pi,\pi'$ two partitions of $N$, where $\pi = (S_1,...,S_k)$ and $\pi'=(T_1,...,T_l)$ with $|S_i| \geq |S_{i+1}| \; \forall i\in \{1,...,k-1\}$ and $|T_j| \geq |T_{j+1}| \; \forall j\in \{1,...,l-1\}$ respectively. We say, that
\begin{align*}
	\noindent
\pi \gdot \pi' &\Leftrightarrow \exists i \leq \min\{k,l\}: |S_i| > |T_i| \; and \; |S_j| = |T_j| \; \forall j<i \quad \\
\&~ \pi \doteq \pi' &\Leftrightarrow k=l \; and \; \forall i\leq k: |S_i| = |T_i|.
\end{align*}
\end{definition}

The relation $\gdot$ is complete, transitive and asymmetric, and places an ordering on the set of 
partitions. 
We now present the proof of Theorem~\ref{th:bottom-sis} in which we utilize the relation $\gdot$.

\begin{proof}\label{proof_br_is}
	To simplify the presentation, we prove that every bottom responsive game admits an IS partition. The same argument can also be used to show that every 
	bottom responsive game admits an SIS partition.
	
We show individual stability for each 
maximum element according to $\gdot$ of the set of individual rational coalitions. 
Consider the set $P = \{ \pi' \midd \pi' \; partitions  \; N \; and \; \forall S \in \pi', i \in S: \{i\} \in Av(i,S) \} $. Note that $P \neq \emptyset$, because the partition consisting of only singletons is in $P$ and that $P$ is a finite set because the number of partitions is finite. Denote by $\pi$ a maximal element of $P$ according to $\gdot$, i.e. $\pi \gdoteq \pi' \; \forall \pi' \in P$. By definition, $\pi$ is individually rational.

Now assume $\pi$ is not individually stable. Then, there exists a player $i \in N $ and a coalition $S \in \pi \cup \{\emptyset\}$, such that $S \cup \{i\} \succ_i \pi(i)$ and $\forall j \in S: S \cup \{i\} \succsim_j S$. Now we show that the partition $\pi$ resulting after the deviation of $i$ is still individually rational and therefore an element of $P$. Clearly $S \neq \emptyset$ because of individual rationality of $\pi$. Furthermore $\{j\} \in Av(j,S\cup \{i\}) \; \forall j \in S$, because if not $  S\cup \{i\} \prec_j S$ for some $j\in S$.\\
Consider a player $j \in \pi(i) \setminus \{i\}$. Due to individual rationality of $\pi$, $\{j\} \in Av(j,\pi(i))$, which implies $T \succsim_j \{j\} \; \forall T \subseteq \pi(i)$ with $j\in T$. So $\pi(i) \setminus \{i\} \succsim_j \{j\}$. All other players $j \in N\setminus (\pi(i) \cup S)$ are not affected by the deviation of $i$ because of the hedonic game setting. Therefore $\pi'$ is individually rational and $\pi' \in P$.

The last step is to show $\pi' \gdot \pi$, which contradicts the maximality of $ \pi$ in $P$. Because player $i$ improves by changing, $|S \cup \{i\}| > |\pi(i)|$ follows from condition 2) of bottom responsiveness . So $(S \cup \{i\}, \pi(i) \setminus \{i\}) \gdot (S, \pi(i))$ and all other coalitions are identical in $\pi$ and $\pi'$. This contradicts $\pi \gdoteq \pi'$ and finishes the proof.
\end{proof} 

The proof also highlights a decentralized way to compute an IS or SIS partition. Start from the partition of singletons and enable arbitrary deviations. For each partition $\pi_k$, the new partition $\pi_{k+1}$ is such that $\pi_{k+1} \gdot \pi_k$. Therefore, in a finite number of deviations, an IS or SIS partition is achieved.


\paragraph{Proof of Theorem~\ref{th:bottom-SNS}}


We now  present the proof of Theorem~\ref{th:bottom-SNS}.

\begin{proof}\label{proof_br_sns}
We show strong Nash stability for each 
maximal element according to $\gdoteq$ of the set of individual rational coalitions (please see Definition~\ref{def:ordering}). 
Consider the set $P = \{ \pi' \midd \pi' \; partitions  \; N \; and \; \forall S \in \pi', i \in S: \{i\} = av(i,S) \} $. Note that $P \neq \emptyset$, because the partition consisting of only singletons is in $P$ and $P$ is a finite set, because the number of partitions is finite. Denote by $\pi$ a maximal element of $P$ according to $\gdoteq$, i.e. $\pi \gdoteq \pi' \; \forall \pi' \in P$. 

Now assume $\pi$ is not strong Nash stable. Then a set of players $H \subseteq N $ and a partition $\pi'$ exist, such that
\begin{align*}
&(1) \quad \pi \stackrel{H}{\longrightarrow} \pi'\\
&(2) \quad \forall i \in H: \pi'(i) \succ_i \pi(i).
\end{align*}
We show that the partition $\pi'$ resulting after the deviation is still individually rational and therefore an element of $P$. 
Clearly $av(i,\pi'(i)) = \{i\} \; \forall i \in H$, because otherwise $\pi'(i) \succ_i \pi(i)$ would not hold. Now consider a player $j$ such that $\pi'(j) \cap H \neq \emptyset$. $\forall i\in H \cap \pi'(j): j\notin av(i,\pi'(j))$. Mutuality implies $i \notin av(j,\pi'(j))$ and therefore $av(j,\pi'(j)) = av(j,\pi(j)) = \{j\}$. All other players $j\in N$ are either not affected by any changes ($\pi(j) = \pi'(j)$) or they are left by some players in $H$ ($\pi'(j) \subset \pi(j)$). In both cases $av(j,\pi'(j)) = av(j,\pi(j)) = \{j\}$, so $\pi'$ is an element of $P$.

The last step is to show $\pi' \gdot \pi$, which contradicts the maximality of $ \pi$ in $P$. Because each player $i\in H$ improves, $|\pi'(i)| > |\pi(i)| \; \forall i \in H$, which follows from condition (iii) of bottom responsiveness. Take the largest coalition $S\in \pi$ such that $S \cap H \neq \emptyset$. Obviously any coalition bigger than $S$ in $\pi$ at least does not get smaller after the deviation, because it contains no players from $H$. Then one of following two cases holds:\\

Case 1: at least one coalition $T\in \pi$ with $|T| > |S|$ gets joined by some player $i\in H$. But then $\pi' \gdot \pi$, since $T$ increases in size and any larger coalition in $\pi$ does not get smaller.

Case 2: If Case 1 does not hold, we know that no coalition in $\pi$ larger than $S$ is joined by a player in $H$ and therefore stays the same. But one player $i\in S \cap H$ is part of a coalition $S'\in \pi'$ with $|S'|>|S|$. Since all coalitions in $\pi$, which are larger than $S$ also exist in $\pi'$, we can again conclude $\pi' \gdot \pi$.

In both cases $\pi' \gdot \pi$ which contradicts the maximality of $\pi$ in $P$ and finishes the proof.
\end{proof}

\section{Existence of stability for specific classes of games}\label{sec:app}

In this section, we highlight some natural subclasses of additive separable hedonic games~\cite[see e.g., ][]{ABS11a, BoJa02a, Gasa11a, Hajd06a} and hedonic games with \B-preferences~\cite[see e.g., ][]{CeRo01a,Hajd06a} which guarantee top responsiveness or bottom responsiveness. Consequently, our existence results in Section~\ref{sec:exists} and an existence result in the literature~\citep{DiSu07a} applies to these settings.


\paragraph{Additive separable hedonic games}

Additive separable hedonic games are one of the most well-studied and natural class of hedonic games~\cite[see e.g., ][]{ABS11a, BoJa02a, Gasa11a, Hajd06a}.
In an \emph{additively separable hedonic game (ASHG)} $(N,\pref)$, each player $i\in N$ has value $v_i(j)$ for player $j$ being in the same coalition as $i$ and if $i$ is in coalition $S\in \mathcal{N}_i$, then $i$ gets utility $\sum_{j\in S\setminus \{i\}}v_i(j)$. For coalitions $S,T\in\mathcal{N}_i$, $S \succsim_i T$ if and only if $\sum_{j\in S\setminus \{i\}}v_i(j) \geq \sum_{j\in T\setminus \{i\}}v_i(j)$. Therefore an ASHG can be represented as $(N,v)$. 

An ASHG is \emph{symmetric} if $v_i(j)=v_j(i)$ for any two players $i,j\in N$ and is \emph{strict} if $v_i(j)\neq 0$ for all $i,j\in N$. 

We now formally introduce two classes of additive separable hedonic games  which also satisfy top responsiveness and bottom responsiveness respectively. Both classes were introduced by \citet{DBHS06a}.
\begin{itemize}
\item An ASGH $(N,v)$ is \emph{appreciation of friends} if for all $i,j\in N$ such that $i\neq j$, the following holds:  $v_i(j)\in \{-1,+n\}$.
\item An ASGH $(N,v)$ is \emph{aversion to enemies} if for all $i,j\in N$ such that $i\neq j$, the following holds:  $v_i(j)\in \{-n,+1\}$.
\end{itemize}

It is clear that `appreciation of friends' and `aversion to enemies' games are ASHGs with strict preferences. 
\citet{SuSu10a} noted that `appreciation of friends' and `aversion to enemies' games satisfy top responsiveness and bottom responsiveness respectively. As a consequence, our main results apply to these games.

\begin{corollary}
There exists an SSNS partition for each symmetric `appreciation of friends' game.
\end{corollary}
\begin{proof}
	An `appreciation of friends' game satisfies top responsiveness. Furthermore, if the game is (additive separable) symmetric, then it also satisfies mutuality with respect to top responsiveness. Then, as a result of Theorem~\ref{th:ssns-top}, we get the corollary.
\end{proof}

\begin{corollary}
	There exists an SIS partition for each `aversion to enemies' game.
\end{corollary}
\begin{proof}
	The statement follows from Theorem~\ref{th:bottom-sis} and the fact that `aversion to enemies' games satisfy bottom responsiveness.
	\end{proof}

\begin{corollary}
There exists an SNS partition for each symmetric `aversion to enemies' game.
\end{corollary}
\begin{proof}
	It is already known that `aversion to enemies' games satisfy bottom responsiveness. Since `aversion to enemies' are additively separable hedonic games with strict preferences, they not only satisfy bottom responsiveness but also strong bottom responsiveness. If `aversion to enemies' have symmetric preferences, then they not only satisfy strong bottom responsiveness but also (bottom responsive) mutuality. Therefore, we can apply Theorem~\ref{th:bottom-SNS} to derive the corollary. 
\end{proof}

\paragraph{\B-hedonic games}

\eat{
\begin{table}
\begin{center}
\begin{tabular}{r|c|c|c|c|}
\multicolumn{1}{r}{}
 &  \multicolumn{1}{c}{a}
 & \multicolumn{1}{c}{b} 
 & \multicolumn{1}{c}{c}
 & \multicolumn{1}{c}{d}\\
\cline{2-5}
a&4&1&2&3 \\
\cline{2-5}
b&1&4&3&2\\
\cline{2-5}
c&2&3&4&1\\
\cline{2-5}
d&3&2&1&4\\
\cline{2-5}
\end{tabular}
\end{center}
\caption{A 4-player \B-hedonic game with preferences, which are strict, have no unacceptability and satisfy symmetry represented by the table $(t_{ij})$. Entry $t_{ij}$ is the rank of player $j$ in $i$'s preference list. 
}
\label{table:B-game}
\end{table}
}

Finally, we show another important subclass of hedonic games called \B-hedonic games~\citep{CeRo01a, Cech08a} satisfies top responsiveness. In \B-hedonic games, players express preferences over players and these preferences over players are naturally extended to preferences over coalitions.
We will assume that $\max_{i}(\emptyset)=\{i\}$.
	In \emph{hedonic games with \B-preferences} (in short \B-hedonic games), for $S,T\in \mathcal{N}_i$, $S \sPref[i] T$ if and only if one of the following conditions hold: 
	
	\begin{enumerate}
	\item for each $s\in \max_{i}(S\setminus \{i\})$ and $t\in \max_{i}(T\setminus \{i\})$, $s\succ_i t$, or
	\item for each $s\in \max_{i}(S\setminus \{i\})$ and $t\in \max_{i}(T\setminus \{i\})$, $s\sim_i t$ and $|S|<|T|$. 
\end{enumerate}

	A \B-hedonic has strict preferences for each $i\in N$ and $j,k\in N$, the following holds: $j\neq k$ $\Rightarrow$  $j\nsim_i k$.
	


Then, we have the following proposition.

\begin{proposition}\label{prop:B-top}
\B-hedonic games with strict preferences satisfy top responsiveness. 
\end{proposition}
\begin{proof}
We show that \B-hedonic games with strict preferences satisfy all the three conditions of top responsiveness.
\begin{enumerate}

\item Firstly, for each $X\in \mathcal{N}_i$, $Ch(i,X)=\{\max_i{X}\cup \{i\}\}$ and thus $|Ch(i,X)|=1$. 

\item For a pair $X,Y\in \mathcal{N}_i$, assume that $ch(i,X)\succ_i ch(i,Y)$. This means that $\{\max_i(X)\}\cup \{i\} \succ_i \{\max_i{Y}\}\cup \{i\}$. Since the best player in $X$ is more preferred by $i$ than the best player in $Y$, then by the definition of \B-hedonic games, $X\succ_i Y$.

\item Finally, for each pair $X,Y\in \mathcal{N}_i$, assume that $ch(i,X)= ch(i,Y)$ and $X\subset Y$. Then, the player most preferred by $i$ in $X$ is the same as the player player most preferred by $i$ in $Y$. Therefore, by the definition of \B-hedonic games, $X\succ_i Y$.
\end{enumerate}



This completes the proof.
\end{proof}

Therefore, as a corollary we get the following statement which was proved by \citet{CeRo01a}.

\begin{corollary}\label{cor:B-hedonic}
	For each \B-hedonic game with strict preferences, a strict core stable partition is guaranteed to exist. 
\end{corollary}
\begin{proof}
	\citet{DiSu07a} showed that for hedonic games satisfying top responsiveness admit a strict core stable partition. Since \B-hedonic games satisfy top responsiveness, they admit a strict core stable partition.
\end{proof}

It will be interesting to see whether there are any natural restrictions on \B-hedonic games with strict preferences such that not only top responsiveness is satisfied but also (top responsive) mutuality is satisfied. In that case, we can apply Theorem~\ref{th:ssns-top} concerning SSNS to \B-hedonic games.

Our demonstrated connection between \B-hedonic games and top responsiveness goes deeper. 
The essential fact behind previous results concerning \B-hedonic games with strict preferences is that they satisfy top responsiveness.
It turns out that the Top Covering Algorithm in \citep{AlRe04a} generalizes the B-STABLE algorithm in \citep{CeRo01a} and in fact Theorems 4.4 and 5.2 in \citep{AlRe04a} imply Theorem 1 and Theorem 2 in \citep{CeRo01a} respectively. This connection seems to have been unnoticed in the literature. 

\section{Conclusions}\label{sec:conc}

To conclude, we tried to paint a clearer picture of the landscape of stability concepts used in coalition formation games. The concepts ranged from standard ones such as the core to recently introduced concepts such as strong Nash stability. The core and strong Nash stability were generalized to strong individual stability and strict strong Nash stability respectively. The basic inclusion relationships between the stability concepts are depicted in Figure~\ref{fig:relations}. Since hedonic games generalize various matching settings, the relations between the stability concepts also hold in matching settings such as two-sided matching, roommate matching etc.

We then examined restrictions on the preferences of agents which guarantee stable outcomes for the new stability concepts. Three main existence results (Theorems 1, 2 and 3) pertaining to top responsiveness and bottom responsiveness were presented. Our results strengthen or complement a number of results in the literature. We also showed that none of our existence results can be extended to a stronger known stability concept. It was seen that the theorems apply to some natural subclasses of hedonic games which have already been of interest among game-theorists. 
It will be interesting to find further applications of our existence results. 

Identifying the impact of preference restrictions on stability also has algorithmic consequences. Recently, hedonic games have attracted research from an algorithmic and computational complexity point of view. There are various algorithmic questions such as checking the existence of and computing stable partitions for different representations of hedonic games~\citep[see e.g., ][]{Cech08a,Hajd06a}. A general framework of preference restrictions and their impact on stability of partitions promises to be useful in devising generic algorithmic techniques to compute stable partitions. For example, we noted that the Top Covering Algorithm in \citep{AlRe04a} generalizes the B-STABLE algorithm in \citep{CeRo01a} by utilizing the insight that \B-hedonic games with strict preferences satisfy top responsiveness. We also mention the following interesting algorithmic questions. For hedonic games represented by individually rational lists of coalitions~\citep{Ball04a}, what is the computational complexity of testing whether the game satisfies top reponsiveness or bottom responsiveness?

Our focus in the paper has been on sufficient conditions which guarantee the existence of stable outcomes. It will be interesting to see what additional conditions are required to ensure uniqueness of stable partitions for different notions of stability.
Finally, characterizing the conditions for the existence of stability remains an open problem. 



\begin{thebibliography}{17}
\providecommand{\natexlab}[1]{#1}
\providecommand{\url}[1]{\texttt{#1}}
\expandafter\ifx\csname urlstyle\endcsname\relax
  \providecommand{\doi}[1]{doi: #1}\else
  \providecommand{\doi}{doi: \begingroup \urlstyle{rm}\Url}\fi

\bibitem[Alcalde and Revilla(2004)]{AlRe04a}
J.~Alcalde and P.~Revilla.
\newblock {Researching with whom? Stability and manipulation}.
\newblock \emph{Journal of Mathematical Economics}, 40\penalty0 (8):\penalty0
  869--887, 2004.

\bibitem[Aziz et~al.(2011{\natexlab{a}})Aziz, Brandt, and Harrenstein]{ABH11b}
H.~Aziz, F.~Brandt, and P.~Harrenstein.
\newblock Pareto optimality in coalition formation.
\newblock In G.~Persiano, editor, \emph{Proceedings of the 4th International
  Symposium on Algorithmic Game Theory (SAGT)}, Lecture Notes in Computer
  Science (LNCS), pages 93--104. Springer-Verlag, 2011{\natexlab{a}}.

\bibitem[Aziz et~al.(2011{\natexlab{b}})Aziz, Brandt, and Seedig]{ABS11a}
H.~Aziz, F.~Brandt, and H.~G. Seedig.
\newblock Optimal partitions in additively separable hedonic games.
\newblock In T.~Walsh, editor, \emph{Proceedings of the 22nd International
  Joint Conference on Artificial Intelligence (IJCAI)}, pages 43--48. AAAI
  Press, 2011{\natexlab{b}}.

\bibitem[Ballester(2004)]{Ball04a}
C.~Ballester.
\newblock {NP}-completeness in hedonic games.
\newblock \emph{Games and Economic Behavior}, 49\penalty0 (1):\penalty0 1--30,
  2004.

\bibitem[Banerjee et~al.(2001)Banerjee, Konishi, and S{\"o}nmez]{BKS01a}
S.~Banerjee, H.~Konishi, and T.~S{\"o}nmez.
\newblock Core in a simple coalition formation game.
\newblock \emph{Social Choice and Welfare}, 18:\penalty0 135--153, 2001.

\bibitem[Bogomolnaia and Jackson(2002)]{BoJa02a}
A.~Bogomolnaia and M.~O. Jackson.
\newblock The stability of hedonic coalition structures.
\newblock \emph{Games and Economic Behavior}, 38\penalty0 (2):\penalty0
  201--230, 2002.

\bibitem[Cechl{\'a}rov{\'a}(2008)]{Cech08a}
K.~Cechl{\'a}rov{\'a}.
\newblock Stable partition problem.
\newblock In \emph{Encyclopedia of Algorithms}, pages 885--888. Springer, 2008.

\bibitem[Cechl{\'a}rov{\'a} and Romero-Medina(2001)]{CeRo01a}
K.~Cechl{\'a}rov{\'a} and A.~Romero-Medina.
\newblock Stability in coalition formation games.
\newblock \emph{International Journal of Game Theory}, 29:\penalty0 487--494,
  2001.

\bibitem[Dimitrov and Sung(2006)]{DiSu06a}
D.~Dimitrov and S.~C. Sung.
\newblock Top responsiveness and {N}ash stability in coalition formation games.
\newblock \emph{Kybernetika}, 42\penalty0 (4):\penalty0 453--460, 2006.

\bibitem[Dimitrov and Sung(2007)]{DiSu07a}
D.~Dimitrov and S.~C. Sung.
\newblock On top responsiveness and strict core stability.
\newblock \emph{Journal of Mathematical Economics}, 43\penalty0 (2):\penalty0
  130--134, 2007.

\bibitem[Dimitrov et~al.(2006)Dimitrov, Borm, Hendrickx, and Sung]{DBHS06a}
D.~Dimitrov, P.~Borm, R.~Hendrickx, and S.~C. Sung.
\newblock Simple priorities and core stability in hedonic games.
\newblock \emph{Social Choice and Welfare}, 26\penalty0 (2):\penalty0 421--433,
  2006.

\bibitem[Gairing and Savani(2011)]{Gasa11a}
M.~Gairing and R.~Savani.
\newblock Computing stable outcomes in hedonic games with voting-based
  deviations.
\newblock In \emph{Proceedings of the 10th International Joint Conference on
  Autonomous Agents and Multi-Agent Systems (AAMAS)}, pages 559--566, 2011.

\bibitem[Gale and Shapley(1962)]{GaSh62a}
D.~Gale and L.~S. Shapley.
\newblock College admissions and the stability of marriage.
\newblock \emph{The American Mathematical Monthly}, 69\penalty0 (1):\penalty0
  9--15, 1962.

\bibitem[Hajdukov{\'a}(2006)]{Hajd06a}
J.~Hajdukov{\'a}.
\newblock Coalition formation games: {A} survey.
\newblock \emph{International Game Theory Review}, 8\penalty0 (4):\penalty0
  613--641, 2006.

\bibitem[Karakaya(2011)]{Kara11a}
M.~Karakaya.
\newblock Hedonic coalition formation games: A new stability notion.
\newblock \emph{Mathematical Social Sciences}, 2011.

\bibitem[Shapley and Scarf(1974)]{ShSc74a}
L.~S. Shapley and H.~Scarf.
\newblock On cores and indivisibility.
\newblock \emph{Journal of Mathematical Economics}, 1\penalty0 (1):\penalty0
  23--37, 1974.

\bibitem[Suzuki and Sung(2010)]{SuSu10a}
K.~Suzuki and S.~C. Sung.
\newblock Hedonic coalition formation in conservative societies.
\newblock Technical Report 1700921, Social Science Research Network, 2010.

\end{thebibliography}

\end{document}